\documentclass[12pt]{article}
\usepackage[pdftex,bookmarksopen=true,bookmarks=true,unicode,setpagesize]{hyperref}
\hypersetup{colorlinks=true,linkcolor=black,citecolor=black}

\usepackage{amssymb, amsmath, amsthm}

\usepackage{xcolor}

\usepackage{mathtools}

\DeclarePairedDelimiter\floor{\lfloor}{\rfloor}

\allowdisplaybreaks
\usepackage{cite}

\newcommand\R{\mathbb{R}}

\newcommand{\vertiii}[1]{{\vert\kern-0.25ex\vert\kern-0.25ex\vert #1
    \vert\kern-0.25ex\vert\kern-0.25ex\vert}}

\newtheorem{theorem}{Theorem}
\newtheorem{corollary}[theorem]{Corollary}

\newtheorem{proposition}[theorem]{Proposition}

\theoremstyle{remark}

\theoremstyle{remark}

\theoremstyle{remark}
\newtheorem{remark}[theorem]{Remark}

\begin{document}

\vspace{-20mm}
\begin{center}{\Large \bf 
Generalized Segal--Bargmann transform for Poisson distribution revisited}
\end{center}

{\large Chadaphorn Kodsueb}\\ 
School of Mathematical Sciences and Geoinformatics, Institute of Science, Suranaree University of Technology, Nakhon Ratchasima, Thailand 30000;\\
e-mail: \texttt{chadaphorn.ko@sut.ac.th}\vspace{2mm}

{\large Eugene Lytvynov}\\ Department of Mathematics, Swansea University, Bay Campus,  Swansea SA1 8EN, U.K.;\\
e-mail: \texttt{e.lytvynov@swansea.ac.uk} (Corresponding Author)\vspace{2mm}

{\small
\begin{center}
{\bf Abstract}

 \end{center}

\noindent
For  $\alpha>0$ and $\sigma > 0$, we consider the following probability distribution on $\alpha\mathbb N_0$: 
$\pi_{\alpha,\sigma}  = \exp \big(- \frac{\sigma}{{\alpha}^2}\big) \sum_{n=0}^{\infty} \frac{1}{n!} \big(\frac{\sigma}{{\alpha}^2}\big)^n {\delta}_{\alpha n}$, where $\delta_y$ denotes the Dirac measure  with mass at $y$. For $\alpha=1$, $\pi_{1,\sigma}$ is the Poisson distribution with parameter $\sigma$. Furthermore, the centered probability distribution $\tilde \pi_{\alpha,\sigma}  = \exp \big(- \frac{\sigma}{{\alpha}^2}\big) \sum_{n=0}^{\infty} \frac{1}{n!} \big(\frac{\sigma}{{\alpha}^2}\big)^n {\delta}_{\alpha n-\sigma/\alpha}$ weakly converges to $\mu_\sigma$ as $\alpha\to0$. Here $\mu_\sigma$ is the Gaussian distribution with mean zero and variance $\sigma$. Let $(c_n)_{n=0}^\infty$ be the monic polynomial sequence that is orthogonal with respect to the measure $\mu_{\alpha,\sigma}$. In particular, for $\alpha=1$, $(c_n)_{n=0}^\infty$ is a sequence of Charlier polynomials. Let $\mathbb F_\sigma(\mathbb C)$  denote the Bargmann space of all entire functions $f(z)=\sum_{n=0}^\infty f_nz^n$ with $f_n \in \mathbb C$ satisfying $ \sum_{n=0}^{\infty} {| f_n |}^2 \, n! \, \sigma^n < \infty$. The generalized Segal--Bargmann transform associated with the measure $\pi_{\alpha,\sigma}$ is a unitary operator  $\mathcal S:L^2(\alpha\mathbb N_0,\pi_{\alpha,\sigma})\to \mathbb F_\sigma(\mathbb C)$ that satisfies $(\mathcal Sc_n)(z)=z^n$ for $n\in\mathbb N_0$. We present some new results related to the operator $\mathcal S$. In particular, we observe how the study of $\mathcal S$ naturally leads  to the normal ordering in the Weyl algebra.
 
\noindent   
 \noindent

 } \vspace{2mm}

\vspace{2mm}


\section{Introduction}

The concept of the Segal--Bargmann transform was developed by I. E. Segal \cite{Segal1,Segal2,Segal3} and V. Bargmann \cite{Bargmann1,Bargmann2,Bargmann3}. In its simplest form, in the one-dimensional case, the Segal--Bargmann transform is the unitary operator $\mathbb S$ between the complex $L^2$-space of the standard Gaussian distribution $\mu$ on $\mathbb R$ and the Bargmann space $\mathbb F(\mathbb C)$ of entire functions on $\mathbb C$. 

Let us recall that $\mathbb F(\mathbb C)$ is a Hilbert space that consists of all entire functions $f(z)=\sum_{n=0}^\infty f_nz^n$ with $f_n\in\mathbb C$ satisfying $\sum_{n=0}^\infty |f_n|^2n!<\infty$, and the inner product of $ f(z) = \sum_{n=0}^{\infty} f_n z^n$ and $g(z) = \sum_{n=0}^{\infty} g_n z^n$ in $\mathbb F(\mathbb C)$ is given by 
$ (f,g)_{\mathbb F(\mathbb C)} =\sum_{n=0}^{\infty} f_n \, \overline{g_n} \, n!$. The Bargmann space $\mathbb F(\mathbb C)$ is a proper subspace of the complex $L^2$-space $L^2(\mathbb C,\nu)$, where $ \nu(dz) = \pi^{-1} \exp( - |z|^2) \, dA(z)$ and 
 $dA(z)=dx\,dy$ ($z=x+iy$) is the Lebesgue measure on $\mathbb C$. The Bargmann space $\mathbb F(\mathbb C)$ provides a realization of the Fock space over $\mathbb C$. 
 
 Let $(h_n)_{n=0}^\infty$ be the the sequence of monic Hermite polynomials that are orthogonal with respect to the measure $\mu$. The Segal--Bargmann transform $\mathbb S:L^2(\mathbb R,\mu)\to\mathbb F(\mathbb C)$ satisfies $(\mathbb Sh_n)(z)=z^n$ for $n\in\mathbb N_0$. The $\mathbb S$ is an integral operator with the integral kernel $\mathbb E(x,z)=\sum_{n=0}^\infty \frac{z^n}{n!}\,h_n(x)=\exp (- z^2/2 - x z)$. For each $z\in\mathbb C$, the function $\mathbb E(\cdot,z)$ is called a coherent state. It is an eigenfunction for the annihilation operator $a^-$ belonging to eigenvalue $z$. Here the operator $a^-$ satisfies $a^-h_n=nh_{n-1}$ ($n\in \mathbb N_0$). Under the Segal--Bargmann  transform $\mathbb S$, the operator $a^-$ goes over to the operator of differentiation in $\mathbb F(\mathbb C)$. The adjoint of the annihilation  operator $a^-$ is the creation operator $a^+$, satisfying $a^+h_n=h_{n+1}$ ($n\in \mathbb N_0$). Then $\mathbb Sa^+\mathbb S^{-1}$ is the operator of multiplication by the variable $z$ in $\mathbb F(\mathbb C)$.  For each function $f\in L^2(\mathbb R,\mu)$, the restriction of the entire function $\mathbb Sf$ to $\mathbb R$ can be written as $(\mathbb Sf)(z)=\int_{\mathbb R}f(x+z)\mu(dx)$ ($z\in\mathbb R$).   

The Segal--Bargmann transform has been similarly defined and studied in the multivariate case. In fact, in his original papers  \cite{Segal1,Segal2,Segal3}, Segal was already interested in the infinite dimensional case.  However, when defining a counterpart of the Bargmann space, Segal used a sequence of Gaussian measures on $\mathbb C^n$ with $n\in\mathbb N$ increasing to $\infty$. 

A study of the Segal--Bargmann transform associated with a Gaussian measure on a (complex) infinite dimensional space was carried out by Y.~M.~Berezansky and Y.~G.~Kondratiev in their monograph \cite[Chapter~2]{BK}, see also the earlier work \cite{Kondratiev}. Essentially at the same time, the Segal--Bargmann transform in the infinite dimensional setting was discussed within white noise analysis under the name of $S$-transform. See e.g.\ the monographs \cite{HKPS,Obata}  and the references therein.  In fact, in the infinite dimensional setting, it was natural to study the Segal--Bargmann transform of spaces of test and generalized functions.

It is well known that the Poisson distribution and the Poisson point process  possess  quite a few properties that are similar to properties of Gaussian measures. For example, both the Poisson process and a Gaussian measure provide a natural unitary isomorphism between their $L^2$-spaces and the Fock space, see e.g.\    \cite{Surgailis}.

A counterpart of the Segal--Bargmann transform for the Poisson distribution was developed by  N. Asai, I. Kubo, and H.-H. Kuo  \cite{Segal-Bargmann}. In fact, even earlier,  Y.-J. Lee and  H.-H. Shih  \cite{LS} developed a generalized Segal--Bargmann transform  for L\'evy processes, which of course include the Poisson  process. 

The present paper deals with a family of discrete probability distributions $\pi_{\alpha,\sigma}$ ($\alpha>0$, $\sigma>0$) defined by 
\begin{equation}\label{cdtesa4wq4}
\pi_{\alpha,\sigma}  = \exp \bigg(- \frac{\sigma}{{\alpha}^2}\bigg) \sum_{n=0}^{\infty} \frac{1}{n!} \bigg(\frac{\sigma}{{\alpha}^2}\bigg)^n {\delta}_{\alpha n} , 
\end{equation}
where $\delta_y$ denotes the Dirac measure  with mass at $y$.
In particular, for $\alpha=1$, $\pi_{1,\sigma}=\pi_\sigma$ is the Poisson distribution with parameter $\sigma$.  Note that the measure $\pi_{\alpha,\sigma}$ is concentrated on the set $\alpha\mathbb N_0$. Furthermore, the centered probability distributions 
\begin{equation}\label{fxdrsre}
\tilde \pi_{\alpha,\sigma}  = \exp \bigg(- \frac{\sigma}{{\alpha}^2}\bigg) \sum_{n=0}^{\infty} \frac{1}{n!} \bigg(\frac{\sigma}{{\alpha}^2}\bigg)^n {\delta}_{\alpha n-\sigma/\alpha}
\end{equation}
 weakly converge to $\mu_\sigma$ as $\alpha\to0$. Here $\mu_\sigma$ is the Gaussian distribution with mean zero and variance $\sigma$. From the viewpoint of quantum physics, the interpolating parameter~$\alpha$ provides a connection between the particle density of an infinite free Bose gas at zero temperature and a free Bose field, see Remark~\ref{ytst5sw5ewwu} below for detail.

 Let $(c_n)_{n=0}^\infty$ be the monic polynomial sequence that is orthogonal with respect to the measure $\mu_{\alpha,\sigma}$. In particular, for $\alpha=1$, $(c_n)_{n=0}^\infty$ is a sequence of Charlier polynomials. The (exponential) generating function of the sequence $(c_n)_{n=0}^\infty$  has the form  
 \begin{equation} \label{tsrear}
\sum_{n=0}^{\infty} \frac{t^n}{n!} \, c_n (z)=\exp \bigg( \frac{z}{\alpha}\,\log (1+t \alpha) - \frac{\sigma t}{\alpha}\bigg).
\end{equation}
In particular, $(c_n)_{n=0}^\infty$  is a Sheffer sequence.

 Let $\mathbb F_\sigma(\mathbb C)$  denote the Bargmann space of all entire functions $f(z)=\sum_{n=0}^\infty f_nz^n$ with $f_n \in \mathbb C$ satisfying $ \sum_{n=0}^{\infty} {| f_n |}^2 \, n! \, \sigma^n < \infty$   and the inner product of $ f(z) = \sum_{n=0}^{\infty} f_n z^n$ and $g(z) = \sum_{n=0}^{\infty} g_n z^n$ in $\mathbb F_\sigma(\mathbb C)$ is given by 
$ (f,g)_{\mathbb F_\sigma (\mathbb C)} =\sum_{n=0}^{\infty} f_n \, \overline{g_n} \, n! \, \sigma^n$. (Thus, the above defined Bargmann space $\mathbb F(\mathbb C)$ corresponds to the choice of the parameter $\sigma=1$ in $\mathbb F_\sigma(\mathbb C)$.)
We define a generalized Segal--Bargmann transform associated with the measure $\pi_{\alpha,\sigma}$ as a unitary operator  $\mathcal S:L^2(\alpha\mathbb N_0,\pi_{\alpha,\sigma})\to \mathbb F_\sigma(\mathbb C)$ that satisfies $(\mathcal Sc_n)(z)=z^n$ for $n\in\mathbb N_0$. In the case $\alpha=1$, this is the operator studied in \cite{Segal-Bargmann}. In this paper, we will discuss some new (together with some old) results related to the operator $\mathcal S$.

The paper is organized as follows. In Section~\ref{cxgtstst}, we will recall several key results of the umbral calculus, which is the  theory of Sheffer polynomial sequences and Sheffer operators, e.g.\ \cite{Rota-paper,11-Kung,Roman}. We also discuss S. Grabiner's result \cite{Grabiner} on an extension of a Sheffer operator (acting on polynomials) to a self-homeomorphism of the space  ${\mathcal E}_{\min}^1 (\mathbb C)$ of entire functions of order at most 1 and minimal type. 

For the reader's convenience, in Section~\ref{xdzewa4wa46}, we will recall some well-known facts about the classical Segal--Bargmann transform (for a Gaussian measure) in the one-dimensional case. We will recall, in particular, how the Segal--Bargmann transform can be interpreted as (an extension of) a Sheffer operator.    

In Section~\ref{gcfcfcdtrd}, we will define and study the  generalized Segal--Bargmann transform $\mathcal S$ for the probability measure $\pi_{\alpha,\sigma}$. For $\alpha=1$, part of the results in this section are due to Asai et al.\  \cite{Segal-Bargmann}. Nevertheless, some results of this section are new even in the Poisson case. For example, we prove that, for $z\in\mathbb R$, $z>-\sigma/\alpha$, $(\mathcal Sf)(z)$ can be written as $(\mathcal Sf)(z)=\int f\,d\pi_{\alpha,\sigma+\alpha z}$. We also show that the operator $\mathcal S$, restricted to polynomials, can written as $\mathcal S=E_{\sigma/\alpha}\mathcal T_\alpha$. Here $E_{\sigma/\alpha}$ is the operator  of shift by $\sigma/\alpha$, and $\mathcal T_\alpha$ is the Sheffer operator associated with a sequence of Touchard polynomials.

Our studies in Section~\ref{gcfcfcdtrd} will naturally lead us to a pair of operators, $\mathcal U$ and $\mathcal V$ that act on polynomials and satisfy the commutation relation  $[\mathcal V,\mathcal U]=\alpha$, hence they are generators of a Weyl algebra. The main connection with the operator $\mathcal S$ is that, under the action of $\mathcal S$, the operator of multiplication by the variable goes over to the operator $\mathcal U\mathcal V$.  As a consequence of the normal ordering in the Weyl algebra \cite{10-Katriel}, we derive explicit formulas for the polynomials $c_n$, as well as an explicit representation of a monomial though the polynomials $c_n$. 

Finally, in Section \ref{fxdzdzsese}, we use Grabiner's result \cite{Grabiner} to treat the operator $\mathcal S$ as a self-homeo\-morphisim of ${\mathcal E}_{\min}^1(\mathbb C)$.  We also study the operators $U=\mathcal S^{-1}\mathcal U\mathcal S$ and $V=\mathcal S^{-1}\mathcal V\mathcal S$, acting in ${\mathcal E}_{\min}^1(\mathbb C)$. We prove that these operators act as follows: $(Uf)(z)=zf(z-\alpha)$, $(Vf)(z)=f(z+\alpha)$.

We note that, in our recent paper \cite{KL}, we dealt  with the generalized Segal--Bargmann transform for the remaining Sheffer sequences of orthogonal polynomials in the classification of Meixner \cite{13-Meixner}.

An extension of the results of this paper to an infinite dimensional setting will be a topic of our future research. 

\section{Elements of umbral calculus}\label{cxgtstst}

\subsection{Sheffer sequences}

This subsection is based on \cite{Rota-paper} and Chapter IV, Sections 3 and 4 of \cite{11-Kung}.

 Let $\mathbb C[z]$ denote  vector space of polynomials over $\mathbb C$. We denote by $\mathcal L(\mathbb C[z])$ the vector space of linear operators acting in $\mathbb C[z]$. 

For $h\in\mathbb C$, we define $E_h\in  \mathcal L(\mathbb C[z])$  as the operator of shift by $h$: $(E_hp)(z)=p(z+h)$ for $p\in\mathbb C[z]$. Boole's formula states that $E_h=e^{hD}=\sum_{n=0}^\infty \frac{h^n}{n!}\,D^n$, where $D\in  \mathcal L(\mathbb C[z])$ is the operator of differentiation. An operator $Q \in \mathcal L(\mathbb C[z])$ is called shift-invariant if $QE_h = E_h Q$ for each $h\in\mathbb{C}$. An operator $Q \in \mathcal L(\mathbb C[z])$  is called a delta operator if $Q$ is  shift-invariant and $Qz = 1$.  

Let $(p_n)_{n=0}^\infty$ be a monic polynomial sequence, i.e., $p_n\in\mathbb C[z]$ ($n\in\mathbb N$), the degree of $p_n$ is $n$, and the coefficient by $z^n$ is 1. For $(p_n)_{n=0}^\infty$, its lowering operator $Q \in \mathcal L(\mathbb C[z])$ is defined by $Qp_n=np_{n-1}$ for $n\in\mathbb N_0$.

A monic polynomial sequence $(p_n)_{n=0}^\infty$ is said to be of binomial type if $p_n(z+\zeta)=\sum_{k=0}^n\binom nk p_k(z)p_{n-k}(\zeta)$ for all $n\in\mathbb N$ and $z,\zeta\in\mathbb C$. 

\begin{theorem} \label{thmBT}
	Let $(p_n)_{n=0}^\infty$ be a monic polynomial sequence, and let $Q$ be its lowering operator.
	The following statements are equivalent:
	
	 (BT1) The sequence $(p_n)_{n=0}^\infty$ is of binomial type.
	
	(BT2) The operator $Q$ is a delta operator.

	(BT3) The operator $Q$ is of the form $ Q= C(D) = \sum_{n=1}^{\infty} c_n D^n$, where $C(t) = \sum_{n=1}^{\infty} c_n t^n$ is a formal power series over $\mathbb C$ with $c_1 = 1$.

	(BT4) The polynomial sequence $(p_n)_{n=0}^\infty$ has the (exponential) generating function of the form 
	\begin{equation}\label{eq:binomial}
		\sum_{n=0}^{\infty} p_n (z)  \frac{t^n}{n!} = \exp(z B(t)),
		\end{equation}
where $B(t) = \sum_{n=1}^{\infty} b_n  t^n$ is a formal power series over $\mathbb C$ with $b_1 = 1$, and formula \eqref{eq:binomial} is understood as an equality of formal power series in~$t$.

Furthermore, the formal power series $B(t)$ and $C(t)$ above are inverse of each other, i.e., $B(C(t)) = C(B(t)) = t$.  
\end{theorem}

By Theorem~\ref{thmBT}, for each delta operator $Q$, there exits a unique binomial sequence $(p_n)_{n=0}^\infty$ for which $Q$ is its lowering operator. Then $(p_n)_{n=0}^\infty$ is called a   basic sequence for $Q$.

The falling factorials are defined by $(z)_0=1$ and
	$$(z) _n = z(z-1) \dotsm (z-n+1),\quad n\in\mathbb N.$$
This is a polynomial sequence of binomial type with generating function of the form~\eqref{eq:binomial}
 in which  $B(t)= \log(1+t)$.  

Note that Stirling numbers of the first kind, $s(n,k)$, are the coefficient of the expansion  
\begin{equation}\label{csresr5as5a}(z) _n = \sum_{k=1}^{n} s(n,k) z^k,
\end{equation} while  Stirling numbers of the second kind, $S(n,k)$, are the coefficient of the expansion  
\begin{equation}\label{fxresera}z^n = \sum_{k=1}^{n} S(n,k) (z)_k.
\end{equation}

One defines the monic polynomial sequence of Touchard (or  exponential) polynomials by 
\begin{equation}\label{cfgxt5w}
T_n (z) = \sum_{k=1}^{n} S(n,k) \, z^k.\end{equation} 
This is a polynomial sequence of binomial type with generating function   of the form
\begin{equation}\label{TouchardGT}
\sum_{n=0}^{\infty} T_n (z)  \frac{t^n}{n!} = \exp\big(z (e^t - 1)\big).
\end{equation}

Let $Q$ be a delta operator.  A monic polynomial sequence $(s_n )_{n=0}^\infty$ is called a Sheffer sequence for $Q$ if $Q$ is the lowering operator for $(s_n)_{n=0}^\infty$.

\begin{theorem} \label{thmSS}
	Let $Q=C(D)$ be a delta operator with basic sequence $(p_n )_{n=0}^\infty$ that has generating function \eqref{eq:binomial}.  
	 Let $(s_n)_{n=0}^\infty$ be a monic polynomial sequence. Then the following statements are equivalent:
	\par (SS1) The polynomial sequence $(s_n )_{n=0}^\infty$ is a Sheffer sequence for $Q$.
	\par (SS2) There exists a (unique, invertible) shift-invariant operator $T$  that satisfies
	 $ s_n (z) = (Tp_n) (z)$ for all $n\in\mathbb N_0$.
	 	\par (SS3) The polynomial sequence $(s_n)_{n=0}^\infty$ has the (exponential) generating function of the form
	\begin{equation}\label{SS3} 
	\sum_{n=0}^{\infty} s_n (z) \, \frac{t^n}{n!} = \exp(z  B(t))  A(t),   \end{equation} 
	where $B(t)$ is as in \eqref{eq:binomial} and 
	$A(t) = \sum_{n=0}^{\infty} a_n t^n$ is a formal power series over $\mathbb C$ with $a_0 = 1$.
	
	Furthermore, $T=\tau(D)=\sum_{k=0}^\infty\tau_k D^k$, where the formal power series $\tau(t)=\sum_{k=0}^\infty \tau_kt^k$ satisfies $\tau(t)=A(C(t))$, where $C(t)$ is the compositional inverse of $B(t)$.

\end{theorem}

A monic polynomial sequence $(s_n )_{n=0}^\infty$ is called an Appell sequence if it is a Sheffer sequence for the differential operator $D$. By Theorem~\ref{thmSS}, a monic polynomial sequence $(s_n )_{n=0}^\infty$ is an Appell sequence if and only if its generating function is of the form \eqref{SS3} with $B(t)=t$. 

\subsection{Sheffer operators}

This section is based on \cite[Chapter~3]{Roman}, \cite{Grabiner} and \cite{FinkelshteinEtAll}.

Let $(s_n)_{n=0}^\infty$ be a (monic) Sheffer sequence. One defines the Sheffer operator associated with the sequence $(s_n)_{n=0}^\infty$ as the operator $S\in\mathcal L(\mathbb C[z])$ that satisfies $Sz^n=s_n(z)$ for all $n\in\mathbb N_0$.  In the case where  $(s_n)_{n=0}^\infty$ is an Appell  sequence, the associated operator $S$ is called an Appell operator. In the case where  $(s_n)_{n=0}^\infty=(p_n)_{n=0}^\infty$ is a polynomial sequence of binomial type, the associated operator $P\in\mathcal L(\mathbb C[z])$ is called an umbral operator. We denote by $\mathfrak S$, $\mathfrak A$, and $\mathfrak B$ the sets of all Sheffer operators, Appell operators, and umbral operators, respectively. 

\begin{theorem}\label{gdtrsw56we}
The $\mathfrak S$ is a group for the product (composition) of linear operators in $\mathbb C[z]$. The $\mathfrak A$ is an abelian normal subgroup of $\mathfrak S$, $\mathfrak B$ is a subgroup of $\mathfrak S$, and $\mathfrak S$ is the semidirect product  of $\mathfrak A$ and $\mathfrak B$, i.e.,  $\mathfrak S=\mathfrak A\rtimes\mathfrak B$. 
\end{theorem}

Theorem~\ref{gdtrsw56we} implies that, if $S$ is a Sheffer operator and $(r_n)_{n=0}^\infty$ is a Sheffer sequence, then the monic polynomial sequence $(Sr_n)_{n=0}^\infty$ is also a Sheffer sequence. A similar statement holds for Appell operators and Appell sequences, as well as umbral operators and sequences of binomial type. 

We will now consider an extension of a class of Sheffer operators to a space of entire functions.

Let $f : \mathbb C \rightarrow \mathbb C$ be an entire function.  One says that $f$ is of order at most 1 and minimal type (when the order is equal to 1) if $f$ satisfies the estimate 
\begin{equation}\label{tstrsara5raq35} \|f\|_n=\sup_{z \in \mathbb C} |f(z)| \exp(-|z|/n) < \infty \quad \text{for all } n\in\mathbb N. 
\end{equation} 
One denotes by ${\mathcal E}_{\min}^1 (\mathbb C)$ the vector space of all such functions.
For each $n\in\mathbb N$, $\lVert \cdot \rVert_{n} $
is a norm  on ${\mathcal E}_{\min}^1 (\mathbb C)$.  Let $B_{n}(\mathbb C)$ denote the Banach space obtained as the completion of  ${\mathcal E}_{\min}^1 (\mathbb C)$ in this norm, i.e., $B_n(\mathbb C)$ is the space of all entire functions $f$ whose norm $\|f\|_n$ is finite. 
 For any $n_1 > n_2$, we have  
$ B_{n_1}(\mathbb C) \subset B_{ n_2}(\mathbb C)$, and the embedding of $B_{n_1}(\mathbb C)$ into $B_{n_2}(\mathbb C)$ is continuous. Note that, as a set, $ {\mathcal E}_{\min}^1 (\mathbb C) = \bigcap_{n=1}^\infty B_{n}(\mathbb C)$. Hence, one defines the projective limit topology on ${\mathcal E}_{\min}^1(\mathbb C)$ given by the norms $\lVert \cdot \rVert_{n}$ ($n\in\mathbb N$). In particular, ${\mathcal E}_{\min}^1 (\mathbb C)$ is a Fr\'echet space.

Let us now discuss an equivalent description of the Fr\'echet space ${\mathcal E}_{\min}^1 (\mathbb C)$. For each $l\in\mathbb N$, denote by $E_l(\mathbb C)$ the Hilbert space of all entire functions $f(z) = \sum_{n=0}^\infty f_n z^n$ that satisfy  
$$ \mathcal N_l(f)=\bigg(\sum_{n=0}^\infty |f_n|^2  (n!)^22^{nl}\bigg)^{1/2} < \infty,$$
 and the norm of $f\in E_l(\mathbb C)$ is $\mathcal N_l(f)$.  Then, as a set, $ {\mathcal E}_{\min}^1 (\mathbb C) = \bigcap_{l=1}^\infty E_{l}(\mathbb C)$, and the topology on $ {\mathcal E}_{\min}^1 (\mathbb C) $ coincides with the projective limit of the $E_l(\mathbb C)$ spaces.    

\begin{theorem}[Grabiner]   \label{relatedSheffer} 
Let $(s_n)_{n=1}^\infty$ be a Sheffer  sequence with generating function~\eqref{SS3} such that $A(t)$ and $B(t)$ are holomorphic functions in a neighborhood of zero. Then the  Sheffer operator $ S \in\mathcal L(\mathbb C[z])$ defined by 
$S z^n = s_n (z)$ ($n\in\mathbb N_0$) extends by continuity to a linear self-homeomorphism of the space $ {\mathcal E}_{\min}^1 (\mathbb C)$. In particular, each function $f \in {\mathcal E}_{\min}^1 (\mathbb C)$ admits a unique representation \begin{equation} \label{uniqueRep}
f(z) = \sum_{n=0}^\infty f_n \, s_n (z) ,
\end{equation}
where the series on the right-hand side of formula \eqref{uniqueRep} converges in ${\mathcal E}_{\min}^1(\mathbb C)$.

\end{theorem}

\begin{corollary} \label{CorrfllwThm}
Let $ (s_n)_{n=0}^\infty$ be a Sheffer sequence satisfying the condition of Theorem~\ref{relatedSheffer}. For each $l\in\mathbb N$, denote by $H_l(\mathbb C)$ the Hilbert space of all entire functions $f(z) = \sum_{n=0}^\infty f_n s_n(z)$ that satisfy  
$${\vertiii{f}}_{l}=\bigg(\sum_{n=0}^\infty |f_n|  (n!)^2 2^{nl}\bigg)^{1/2} < \infty ,$$
 and the norm of $f\in H_l(\mathbb C)$ is ${\vertiii{f}}_{l}$.  Then, as a set, $ {\mathcal E}_{\min}^1 (\mathbb C) = \bigcap_{l=1}^\infty H_{l}(\mathbb C)$, and the topology on $ {\mathcal E}_{\min}^1 (\mathbb C) $ coincides with the projective limit of the $H_l(\mathbb C)$ spaces.    
\end{corollary}

\section{Segal--Bargmann transform in the one-dimensional case}\label{xdzewa4wa46}

In this section, we follow \cite{Bargmann1} and \cite[Chapter~2, Section~5.2]{BK} (in a slightly generalized form).

Let $\sigma > 0$. Let $\nu_\sigma$ be the Gaussian measure on $\mathbb C$ given by 
$$ \nu_\sigma (dz) = \frac{1}{\pi\sigma} \, \exp\left( - \frac{|z|^2}{\sigma} \right) \, dA(z).$$
(Recall that  $dA(z)=dx\,dy$ for $z=x+iy$.) It holds that
 $$(z^m , z^n)_{L^2 (\mathbb C , \nu_\sigma)}=\int_{\mathbb C} z^m \, \overline{z^n} \, \nu_\sigma (dz) = \delta_{m,n} \, n! \, \sigma^n,\quad m, n \in {\mathbb N}_0,$$
where $\delta_{m,n}$ is the Kronecker delta. 
 
 Recall the Bargmann spce $\mathbb F_\sigma (\mathbb C)$ that was defined in the Introduction.
 Note that, for $f,g\in\mathbb F_\sigma(\mathbb C)$, we have
$  (f,g)_{\mathbb F_\sigma (\mathbb C)} =(f,g)_{L^2(\mathbb C,\nu_\sigma)}$.
Thus,  $\mathbb F_\sigma (\mathbb C)$  is a proper subspace of the Hilbert space $L^2(\mathbb C,\nu_\sigma)$. The  $\mathbb F_\sigma (\mathbb C)$ is called a Bargmann space. 

Let $m\in\mathbb R$, and let $\mu_{m,\sigma}$ denote  the Gaussian measure on $\mathbb R$ with mean $m$ and variance $\sigma$. 
For $m=0$, we denote $\mu_\sigma=\mu_{0,\sigma}$. 
Let $ (h_n)_{n=0}^\infty $ be the sequence of monic Hermite polynomials that are orthogonal with respect to the measure $\mu_\sigma$\,. The Hermite polynomials $(h_n)_{n=0}^\infty$ satisfy the recurrence relation 
\begin{equation}\label{gfxsresa4ea4}
	z h_n (z) = h_{n+1} (z) + \sigma n h_{n-1} (z), \quad n\in\mathbb N_0.			
\end{equation}
The (exponential) generating function of $ (h_n)_{n=0}^\infty $ has the form  
\begin{equation} \label{HermiteGT}
\sum_{n=0}^{\infty} \frac{t^n}{n!} \, h_n (z) = \exp \bigg( z t - \frac{1}{2} \sigma t^2 \bigg) .
\end{equation} 
In particular, $(h_n)_{n=0}^\infty$ is an Appell sequence. It follows from \eqref{gfxsresa4ea4} that $\|h_n\|_{L^2(\mathbb R,\mu_\sigma)}^2=n!\,\sigma^n$.

For each $z\in\mathbb C$, the corresponding coherent state is defined by
 \begin{equation} \label{Eo}
\mathbb E(x, z)= \sum_{n=0}^{\infty} \frac{z^n}{n! \, \sigma^n} \, h_n (x),\quad x\in\mathbb R.
\end{equation}
Thus, for each $z\in\mathbb C$, we have $\mathbb E(\cdot,z)\in L^2(\mathbb R,\mu_\sigma)$, and for each $x\in\mathbb R$, $\mathbb E(x,\cdot)\in \mathbb F_\sigma(\mathbb C)$. 
By \eqref{HermiteGT} and \eqref{Eo}, 
\begin{equation}\label{vcser5a4q33}
\mathbb E(x,z)= \exp \bigg( - \frac{z^2 - 2x z}{2\sigma} \bigg),\quad x\in\mathbb R,\ z\in\mathbb C.
\end{equation}

The Segal--Bargmann transform is the unitary operator
$\mathbb S : L^2 (\mathbb R , \mu_\sigma) \rightarrow \mathbb F_\sigma (\mathbb C)
$
satisfying 
$ (\mathbb S h_n)(z) = z^n$.  Thus, for $f\in L^2(\mathbb R,\mu_\sigma)$ and $\varphi\in\mathbb F_\sigma(\mathbb C)$, 
\begin{align} \label{Sform}
(\mathbb S f)(z) &= \int_{\mathbb R} f(x) \mathbb E(x, z) \, \mu_\sigma(dx),\quad z\in\mathbb C,\\
(\mathbb S^{-1} \varphi)(x) &= \int_{\mathbb C} \varphi(z) \mathbb E(x, z) \, \nu_\sigma(dz),\quad x\in\mathbb R.\notag
\end{align}

By \eqref{vcser5a4q33} and \eqref{Sform}, 
\begin{align}
(\mathbb S f)(z)&= \int_{\mathbb R} f(x)  \exp \bigg( - \frac{z^2 -2xz}{2\sigma} \bigg)\,\frac{1}{\sqrt{2\pi\sigma}} \, \exp \bigg( - \frac{x^2}{2\sigma} \bigg) \, dx \notag\\
&= \int_{\mathbb R} f(x)  \frac{1}{\sqrt{2\pi\sigma}} \, \exp \bigg( - \frac{(x -z)^2}{2\sigma} \bigg) \, dx. \label{trsd5ws}
\end{align} 
Hence, 
\begin{equation}\label{vcxtesr4y5}
(\mathbb S f)(z) = \int_{\mathbb R} f(x + z) \, \mu_\sigma (dx)  = \int_{\mathbb R} f(x)  \,\mu_{z,\sigma} (dx),\quad z \in \mathbb R.\end{equation}
Note that, for $p\in\mathbb C[z]\subset L^2(\mathbb R,\mu_\sigma)$, formula \eqref{vcxtesr4y5} implies 
\begin{equation}\label{bhfytdy6i} 
(\mathbb S p)(z) = \int_{\mathbb R} p(x + z) \, \mu_\sigma (dx),\quad z\in\mathbb C.
\end{equation}

It will be useful for us below to give another interpretation of formula \eqref{trsd5ws}. For each $z\in\mathbb C\setminus\mathbb R$, we define the complex-valued measure $\mu_{z,\sigma} (dx)$ on $\mathbb R$
 by 
$$\mu_{z,\sigma} (dx)=\frac{1}{\sqrt{2\pi\sigma}} \, \exp \bigg( - \frac{(x -z)^2}{2\sigma} \bigg) \, dx.$$
This measure $\mu_{z,\sigma}$ can be thought of as a complex-valued Gaussian measure on $\R$. 

Then, by \eqref{trsd5ws}, for each $f\in L^2(\mathbb R,\mu_\sigma)$, we get
$$
(\mathbb S f)(z) = \int_{\mathbb R} f  \,d\mu_{z,\sigma},\quad z\in\mathbb C.$$
Note that $ \mu_{z,\sigma} (\mathbb R) = {\mathbb S} 1 = 1$.

Since both the Hermite sequence $ (h_n)_{n=0}^\infty $ and the sequence of monomials $(z^n)_{n=0}^\infty$ are Appell systems, Theorem~\ref{gdtrsw56we} implies that the restriction of $\mathbb S$ to $\mathbb C[z]$ is the Appell operator associated with the Appell sequence $(\tilde h_n)_{n=0}^\infty$, where $\tilde h_n(z)=\mathbb S z^n$. It follows from \eqref{bhfytdy6i}, the explicit form of the moments of the measure $\mu_\sigma$, and the explicit form of the Hermite polynomials~$h_n$ that 
\begin{align}
\tilde h_n(z)&=\int_{\mathbb R}(x+z)^k\,\mu_\sigma(dx)= 
\sum_{k=0}^n \binom n k z^{n-k}\int_{\mathbb R} x^{k}\,\mu_\sigma(dx)\notag\\
&=z^n+\sum_{m=1}^{\floor*{\frac n2}}\binom n{2m}z^{n-2m}\int_{\mathbb R}x^{2m}\,\mu_{\sigma}(dx)
=z^n+n!\sum_{m=1}^{\floor*{\frac n2}}z^{n-2m}\, \frac{\sigma^m}{2^m(n-2m)!\,m!}\notag\\
&=i^n\bigg((-iz)^n+n!\sum_{m=1}^{\floor*{\frac n2}}(-iz)^{n-2m}(-1)^m\, \frac{\sigma^m}{2^m(n-2m)!\,m!}\bigg)=i^nh_n(-iz).\label{fstdtqdstq}
\end{align}
By \eqref{HermiteGT} and \eqref{fstdtqdstq},
$$\sum_{n=0}^{\infty} \frac{t^n}{n!} \, \tilde h_n (z) = \exp \bigg( z t + \frac{1}{2} \sigma t^2 \bigg) . $$
 
 Denote by $a^+$ and $a^-$ the raising and lowering operators for the monic polynomial sequence $(h_n)_{n=0}^\infty$, i.e., $a^+, a^-\in\mathcal L(\mathbb C[z])$ and  
 \begin{equation} 	\label{zcx876-=rhgfd}
	(a^+ h_n) (z) = h_{n+1} (z), \quad (a^-h_n)(z )= n h_{n-1} (z), \quad\quad n \in \mathbb N_0 .
\end{equation}
(These operators are often called the creation and annihilation operators, respectively.)

Let $Z\in\mathcal L(\mathbb C[z])$ denote the operator of multiplication by the variable $z$. By \eqref{gfxsresa4ea4} and \eqref{zcx876-=rhgfd}, $Z=a^++\sigma a^-$. The operator $Z$ is essentially self-adjoint in $L^2(\mathbb R,\mu_\sigma)$ and its closure is the operator of multiplication by the variable in $L^2(\mathbb R,\mu_\sigma)$.

Next, it is easy to see that the operator $a^-$ is closable in $L^2(\mathbb R,\mu_\sigma)$ and we keep the notation $a^-$ for its closure. Then, for each $z \in \mathbb C$, the coherent state $\mathbb E (\cdot, z)$ is an eigenvector of the operator $\sigma  a^-$ belonging to the eigenvalue $z$.

Obviously, 
$$\mathbb Sa^+\mathbb S^{-1}=Z,\quad \mathbb Sa^-\mathbb S^{-1}=D,$$
so that
$$\mathbb SZ\mathbb S^{-1}=Z+\sigma D.$$
The operator $Z+\sigma D\in\mathcal L(\mathbb C[z])$ is essentially self-adjoint in $\mathbb F_\sigma(\mathbb C)$.

\section{A generalized Segal--Bargmann transform for $\pi_{\alpha,\sigma}$}\label{gcfcfcdtrd}
	
		For  $\alpha>0$ and $\sigma > 0$, let the probability measure $\pi_{\alpha,\sigma}$ on $\alpha\mathbb N_0$ be defined by \eqref{cdtesa4wq4}. Let $(c_n)_{n=0}^\infty$ denote the monic polynomial sequence that is orthogonal with respect to~$\pi_{\alpha,\sigma}$. 	  The $(c_n)_{n=0}^\infty$ is a Sheffer sequence having the generating function
\begin{equation} \label{tsrear}
\sum_{n=0}^{\infty} \frac{t^n}{n!} \, c_n (z)=\exp \bigg( \frac{z}{\alpha}\,\log (1+t \alpha) - \frac{\sigma t}{\alpha}\bigg).
\end{equation}
The sequence $(c_n)_{n=0}^\infty$   satisfies the recurrence formula 
\begin{equation} \label{rcrCharlier} 
z s_n(z) = c_{n+1} (z) + (\alpha n + \sigma/\alpha ) c_n (z) + \sigma n  c_{n-1} (z),\quad n\in\mathbb N_0 . \end{equation} 
It follows from \eqref{rcrCharlier}  that $\|c_n\|_{L^2(\alpha\mathbb N_0,\pi_{\alpha,\sigma})}^2=n!\,\sigma^n$. 
		
\begin{remark}\label{rtstsw5w456}
Let $\tilde \pi_{\alpha,\sigma}$ be the centered measure  $\pi_{\alpha,\sigma}$, see 
  \eqref{fxdrsre}.  The Fourier transform of $\tilde \pi_{\alpha,\sigma}$ is
$$\int_{\alpha\mathbb N_0}e^{ixy}\,\tilde\pi_{\alpha,\sigma}(dy)=\exp\bigg(\frac\sigma{\alpha^2}(e^{i\alpha y}-1-i\alpha y)\bigg).$$
Hence,
\begin{equation}\label{crtse5ysw5ue64}
\lim_{\alpha\to0}\int_{\alpha\mathbb N_0}e^{ixy}\,\tilde\pi_{\alpha,\sigma}(dy)=\exp\bigg(-\frac\sigma2\,y^2\bigg)=\int_{\mathbb R}e^{ixy}\,\mu_\sigma(dy).\end{equation}
Thus,   $\tilde \pi_{\alpha,\sigma}$ converges weakly to $\mu_\sigma$ as $\alpha\to0$. 
\end{remark}

\begin{remark}\label{ytst5sw5ewwu}
Formula \eqref{crtse5ysw5ue64} admits the following interpretation from the viewpoint of quantum physics.  Denote by $Z$ the  (unbounded) operator of multiplication by the variable in $L^2(\alpha\mathbb N_0,\pi_{\alpha,\sigma})$: $(Zf)(z)=zf(z)$ for $f(z)$ from the domain of $Z$. 
Consider the complex space $\ell_2$ with its standard orthonormal basis $(e_n)_{n\in\mathbb N_0}$. Here  $e_n=(0,\dots,0,1,0,0,\dots)$, where $1$ is at the $n$th place. Define a unitary operator\linebreak $I:L^2(\alpha\mathbb N_0,\pi_{\alpha,\sigma})\to\ell_2$ satisfying  $Ic_n=(n!\,\sigma^n)^{1/2}e_n$ ($n\in\mathbb N_0$).
Next, define a self-adjoint (unbounded) linear operator $\rho_{\alpha,\sigma}=IZI^{-1}$ in $\ell_2$. Formula \eqref{rcrCharlier} implies 
\begin{equation}\label{fxreazrea4rwa}
\rho_{\alpha,\sigma}\,e_n=\sqrt{\sigma(n+1)}\,e_{n+1}+(\alpha n+\sigma/\alpha)e_n+\sqrt{\sigma n}\,e_{n-1},\quad n\in\mathbb N_0.\end{equation}
Consider  a creation operator $a^+$ and an annihilation operator $a^-$ in $\ell_2$  satisfying $a^+e_n=\sqrt{n+1}\,e_{n+1}$ and $a^-e_n=\sqrt n\,e_n$ ($n\in\mathbb N_0$). The operators $a^+$ and $a^-$ are adjoint of each other and satisfy the commutation relation $[a^-,a^+]=1$. By \eqref{fxreazrea4rwa},
$$\rho_{\alpha,\sigma}=\alpha(a^++\sqrt\sigma/\alpha)(a^-+\sqrt\sigma/\alpha).$$
Note that the operators $A_{\alpha,\sigma}^+=a^++\sqrt\sigma/\alpha$ and $A_{\alpha,\sigma}^-=a^-+\sqrt\sigma/\alpha$ are also adjoint of each other and satisfy the commutation relation $[A_{\alpha,\sigma}^-,A_{\alpha,\sigma}^+]=1$. 
It follows from \cite{AW} (see also \cite{GGPS}) that the operators $A_{\alpha,\sigma}^+$, $A_{\alpha,\sigma}^-$ form a representation of the canonical commutation relations (CCR) describing an infinite free Bose gas at zero temperature.   Hence,  the operator $A_{\alpha,\sigma}^+A_{\alpha,\sigma}^-$ is the particle density of this gas with 
 average density $\sigma/\alpha^2$. We observe that
\begin{align*}
\rho_{\alpha,\sigma}-\frac\sigma\alpha&=\alpha\bigg(A_{\alpha,\sigma}^+A_{\alpha,\sigma}^--\frac\sigma{\alpha^2}\bigg)\\&=\sqrt\sigma (a^++a^-)+\alpha a^+a^-\to \sqrt\sigma(a^++a^-)\quad\text{as }\alpha\to0.
\end{align*}
The limiting operator $\sigma(a^++a^-)$ describes a free Bose field. 

\end{remark}

We define a generalized Segal--Bargmann transform $\mathcal S:L^2(\alpha\mathbb N_0,\pi_{\alpha,\sigma})\to \mathbb F_\sigma(\mathbb C)$ as a unitary operator that satisfies $(\mathcal Sc_n)(z)=z^n$ for $n\in\mathbb N_0$.  

For each $z\in\mathbb C$, the corresponding coherent state is given by
 \begin{equation} \label{cfdyrdstrswt6}
\mathcal E(x, z)= \sum_{n=0}^{\infty} \frac{z^n}{n! \, \sigma^n} \, c_n (x),\quad x\in\alpha\mathbb N_0,\ z\in\mathbb C.
\end{equation}
For each $z\in\mathbb C$, we have $\mathcal E(\cdot,z)\in L^2(\alpha\mathbb N_0,\pi_{\alpha,\sigma})$, and for each $x\in\alpha\mathbb N_0$, $\mathcal E(x,\cdot)\in \mathbb F_\sigma(\mathbb C)$. 
By \eqref{tsrear} and \eqref{cfdyrdstrswt6}, 
\begin{equation}\label{vcser5a4q3}
\mathcal E(\alpha n,z)=\left( 1 + \frac{\alpha z}{\sigma} \right)^{n} \exp \left(- \frac{z}{\alpha} \right),\quad n\in\mathbb N_0,\ z\in\mathbb C.
\end{equation}

Thus, for $f\in L^2(\alpha\mathbb N_0,\pi_{\alpha,\sigma})$ and $\varphi\in\mathbb F_\sigma(\mathbb C)$, 
\begin{align} \label{rtd6d6we6m}
(\mathcal S f)(z)& = \int_{\alpha\mathbb N_0} f(x) \mathcal E(x, z) \, \pi_{\alpha,\sigma}(dx),\quad z\in\mathbb C,\\
(\mathcal S^{-1} \varphi)(x) &= \int_{\mathbb C} \varphi(z) \mathcal E(x, z) \, \nu_\sigma(dz),\quad x\in\alpha\mathbb N_0.\notag
\end{align}

For  $\alpha>0$ and $z\in\mathbb C$, we define a complex-valued measure $\pi_{\alpha,z}$ on $\alpha\mathbb N_0$ by 
$$\pi_{\alpha,z}  = \exp \bigg(- \frac{z}{{\alpha}^2}\bigg) \sum_{n=0}^{\infty} \frac{1}{n!} \bigg(\frac{z}{{\alpha}^2}\bigg)^n {\delta}_{\alpha n} . $$
Thus, for $z>0$, $\pi_{\alpha,z}$ is the probability distribution defined in \eqref{cdtesa4wq4}. Note that, for $z\in(-\infty,0)$, $\pi_{\alpha,z}$ is a signed (real-valued) measure.

\begin{theorem}\label{cfyzdtrdZSTA}
For each $f\in L^2(\alpha\mathbb N_0,\pi_{\alpha,\sigma})$, we have
\begin{equation}\label{ydtrst}
(\mathcal Sf)(z)= \int_{\alpha \mathbb N_0} f \, d\pi_{\alpha,\sigma+\alpha z},\quad z\in\mathbb C.
\end{equation}
In particular, for $z>-\frac\sigma\alpha$, the integration (summation) in \eqref{ydtrst} is with respect to the probability distribution $\pi_{\alpha,\sigma+\alpha z}$. The complex-valued series on the right-hand side of~\eqref{ydtrst} converges absolutely and uniformly on compact sets in $\mathbb C$.  
\end{theorem}

\begin{remark}
For $\alpha=1$, formula  \eqref{ydtrst} was mentioned (without proof) in \cite{KL}. 
\end{remark}

\begin{proof}[Proof of Theotrem \ref{cfyzdtrdZSTA}]
By \eqref{vcser5a4q3} and \eqref{rtd6d6we6m}, for $f \in  L^2 (\alpha {\mathbb N}_0, \pi_{\alpha,\sigma} )$,
\begin{align*}
(\mathcal S f)(z) &= \sum_{n=0}^{\infty} f(\alpha n) \left( 1 + \frac{\alpha z}{\sigma} \right)^n \exp \left(- \frac{z}{\alpha} \right) \exp \left(- \frac{\sigma}{\alpha^2} \right) \frac{1}{n!} \left(\frac{\sigma}{\alpha^2} \right)^n \\
&= \exp \left(-\frac{\sigma+\alpha z}{\alpha^2} \right) \sum_{n=0}^{\infty} f(\alpha n) \, \frac{1}{n!} \left(\frac{\sigma+\alpha z}{\alpha^2} \right)^n= \int_{\alpha \mathbb N_0} f \, d\pi_{\alpha,\sigma+\alpha z}. 
\end{align*}
Next, by the Cauchy  inequality,
\begin{align*}
&\sum_{n=0}^{\infty} |f(\alpha n)| \, \frac{1}{n!} \left|\frac{\sigma+\alpha z}{\alpha^2}  \right|^n \le
\sum_{n=0}^{\infty} |f(\alpha n)| \, \frac{1}{n!} \left(\frac{\sigma+|\alpha z|}{\alpha^2}  \right)^n \\
&\quad=\sum_{n=0}^{\infty} |f(\alpha n)| \, \left(\frac{1}{n!}\right)^{\frac12} \left(\frac{\sigma}{\alpha^2}\right)^{\frac n2}\cdot†
\left(\frac{1}{n!}\right)^{\frac12} 
\left(\frac{\sigma+|\alpha z|}{\alpha^2}  \right)^n \left(\frac{\sigma}{\alpha^2}\right)^{-\frac n2}\\
&\quad\le\left(
\sum_{n=0}^{\infty} |f(\alpha n)|^2\,\frac1{n!}\,\left(\frac{\sigma}{\alpha^2}\right)^{n}
\right)^{\frac12}
\left(
\sum_{n=0}^{\infty} \frac1{n!}\,\bigg(\frac{(\sigma+|\alpha z|)^2}{\sigma}\right)^{n}
\bigg)^{\frac12}\\
&=\exp\bigg(\frac{\sigma}{\alpha^2}\bigg)\|f\|_{L^2(\alpha\mathbb N_0,\pi_{\alpha,\sigma})}\left(
\sum_{n=0}^{\infty} \frac1{n!}\,\left(\frac{(\sigma+|\alpha z|)^2}{\sigma}\right)^{n}
\right)^{\frac12}.
\end{align*}
Hence, the complex-valued series on the right-hand side of \eqref{ydtrst} converges absolutely and uniformly on compact sets in $\mathbb C$.
\end{proof}

Define the monic polynomial sequence $(T_{\alpha,n})_{n=0}^\infty$ by 
\begin{equation} \label{TouchardPln} T_{\alpha, n} (z)=\sum_{k=1}^{n} S(n,k)\alpha^{n-k} z^k = \alpha^n T_n \bigg( \frac{z}{\alpha} \bigg), \quad n\in\mathbb N,
\end{equation} 
where $(T_n )_{n=0}^\infty$ is the sequence of Touchard polynomials, see \eqref{cfgxt5w}. 
By \eqref{TouchardGT}, $(T_{\alpha,n})_{n=0}^\infty$ is a polynomial sequence of binomial type with generating function $$ \sum_{n=0}^{\infty} \frac{t^n}{n!} \, T_{\alpha ,n} (z) = \exp \bigg( z\, \frac{e^{t\alpha} - 1}{\alpha} \bigg). $$
Let $\mathcal T_\alpha$ denote the umbral operator associated with the binomial sequence $(T_{\alpha,n})_{n=0}^\infty$, i.e., $\mathcal T_\alpha z^n=T_{\alpha,n}(z)$ ($n\in\mathbb N_0$). 

Let $((\cdot 
	\mid \alpha)_n)_{n=0}^\infty$ denote the sequence of generalized factorials with increment $\alpha$, i.e., for $z\in\mathbb C$, $(z \mid \alpha)_0 = 1$ and 
	\begin{equation}\label{cyrtdstevf}
			(z\mid \alpha)_n = z(z- \alpha)(z-2\alpha) \dotsm (z-(n-1)\alpha),\quad n\in\mathbb N.\end{equation}
						In particular, $(z\mid1)_n=(z)_n$ is a falling factorial. The $((\cdot \mid \alpha)_n)_{n=0}^\infty$
						is a binomial sequence with generating function
						$$\sum_{n=0}^{\infty} \frac{t^n}{n!} \, (z\mid \alpha)_n = \exp\left(\frac z\alpha \log (1+\alpha t)\right).$$
By \eqref{csresr5as5a} and \eqref{cyrtdstevf}, 
\begin{equation}\label{cdxsedaseaewa}
(z\mid\alpha)_n=\sum_{k=1}^n s(n,k)\alpha^{n-k}z^k, \quad n\in\mathbb N.
\end{equation}

Let $\mathcal F_\alpha$ denote the umbral operator associated with  $((\cdot 
	\mid \alpha)_n)_{n=0}^\infty$, i.e., 
	$$\mathcal F_\alpha z^n=(z\mid\alpha)_n,\quad n\in\mathbb N_0.$$ 
Using  formulas \eqref{csresr5as5a},  \eqref{fxresera}, \eqref{TouchardPln} and  \eqref{cdxsedaseaewa}, one can easily show that $\mathcal F_\alpha$ is the inverse of the operator $\mathcal T_\alpha$, i.e., $\mathcal F_\alpha=\mathcal T_\alpha^{-1}$. 

Note that the restriction of $\mathcal S$ to $\mathbb C[z]$ is a bijective operator in $\mathbb C[z]$, for which we keep the notation $\mathcal S$,

\begin{proposition} \label{bvcsrea5rw4qy} We have the following equalities of linear operators in $\mathbb C[z]$:
$$\mathcal S=E_{\sigma/\alpha}\mathcal T_\alpha,\quad \mathcal S^{-1}=\mathcal F_\alpha E_{-\sigma/\alpha}.$$
Here $E_h$ denotes the operator of shift by $h$.
\end{proposition}

\begin{proof} We use ideas similar to those in \cite[Section~4]{KL}.
We define linear operator $\mathcal U$ and $\mathcal V$ in $\mathbb C[z]$ by 
\begin{equation}\label{cxdsarw4ry}
\mathcal U=Z+\frac\sigma\alpha,\quad \mathcal V= \alpha D+1.
\end{equation}
 As easily seen, the operators $\mathcal U$, $\mathcal V$ satisfy the commutation relation $[\mathcal V,\mathcal U]=\alpha$. Hence, they are generators of a Weyl algebra, see e.g.\ \cite[Section~5.6]{12-Mansour}.  

Let $$ \rho = \mathcal U\mathcal V = Z + \alpha Z D + \frac{\sigma}{\alpha} + \sigma D. $$ Hence, 
  \begin{equation} \label{cfxsttestxhxgg} \rho z^n = z^{n+1} + \left(\alpha n + \frac{\sigma}{\alpha}\right) z^n + \sigma n z^{n-1}. \end{equation} 
  By \eqref{rcrCharlier}  and \eqref{cfxsttestxhxgg}, we have the following equality of linear operators in $\mathbb C[z]$: 
  \begin{equation}\label{cxsrsareare5}\rho=\mathcal S Z\mathcal S^{-1}.
  \end{equation}
    Therefore, 
 \begin{equation}\label{fxrsrtest5u5}\mathcal S z^n=(\rho^n1)(z).\end{equation} 
  
  By Katriel's theorem  \cite{10-Katriel},
\begin{equation}\label{vcftsxreasre}
 \rho^n =(\mathcal U\mathcal V)^n =\sum_{k=1}^{n} S(n,k)\, {\alpha}^{n-k} \mathcal U^k \, \mathcal V^k .\end{equation}
 For each $k\in\mathbb N$,
\begin{equation}\label{xdtarw4}
(\mathcal U^k  \mathcal V^k1)(z)=(\mathcal U^k 1)(z)= \bigg(z+\frac\sigma\alpha\bigg)^k. \end{equation}
By \eqref{TouchardPln} and \eqref{fxrsrtest5u5}--\eqref{xdtarw4}, we get
 \begin{equation}\label{vctrssw5y3w}
 \mathcal Sz^n=\sum_{k=1}^{n} S(n,k)\, {\alpha}^{n-k} \bigg(z+\frac\sigma\alpha\bigg)^k=T_{\alpha,n}\bigg(z+\frac\sigma\alpha\bigg),
 \end{equation}
  which implies $\mathcal S=E_{\sigma/\alpha}\mathcal T_\alpha$. Since $\mathcal F_\alpha=\mathcal T_\alpha^{-1}$, we obtain $\mathcal S^{-1}=\mathcal F_\alpha E_{-\sigma/\alpha}$.  
 \end{proof}

\begin{corollary}\label{fxtrstesw55u} We have, for each $n\in\mathbb N$, 
\begin{equation} \label{ThmCharlier1} z^n = T_{\alpha,n}(\sigma/\alpha)+\sum_{i=1}^n\bigg(\sum_{k=i}^n \binom nk T_{\alpha,n-k}(\sigma/\alpha)S(k,i) \alpha^{k-i} \bigg) c_i (z) , \end{equation} 
and 
\begin{align}  c_n (z) &= \sum_{k=0}^n \binom{n}{k} \bigg(- \frac{\sigma}{\alpha} \bigg)^{n-k} (z \mid \alpha)_k \label{dzawea4tqy44} \\ &= \bigg(- \frac{\sigma}{\alpha} \bigg)^n + \sum_{i=1}^n \bigg( \sum_{k=0}^{n-i} \binom{n}{k} \, s(n-k,i) \, \alpha^{n-2k-i} \, (-1)^k \sigma^k \bigg) z^i .  \label{ThmCharlier3} \end{align}
\end{corollary}

\begin{proof} We have, for $n\in\mathbb N$,
\begin{align*}c_n(z)&=\mathcal S^{-1} z^n= \mathcal F_\alpha E_{-\sigma/\alpha}z^n\\
&=\sum_{k=0}^n\binom nk \bigg(-\frac\sigma\alpha\bigg)^{n-k}\mathcal F_\alpha z^k=\sum_{k=0}^n\binom nk \bigg(-\frac\sigma\alpha\bigg)^{n-k}(z\mid\alpha)_k\,, 
\end{align*}
which implies \eqref{dzawea4tqy44}. Formula \eqref{ThmCharlier3} follows from \eqref{cdxsedaseaewa} and \eqref{dzawea4tqy44}.

Since $(T_{\alpha ,n})_{n=0}^\infty$  is a polynomial sequence of binomial type, formulas~\eqref{TouchardPln} and~\eqref{vctrssw5y3w} imply, for $n\in\mathbb N$, 
\begin{align}
 \mathcal Sz^n&=\sum_{k=0}^n\binom nk T_{\alpha,n-k}(\sigma/\alpha)
T_{\alpha,k}(z)\notag\\
 &=T_{\alpha,n}(\sigma/\alpha)+\sum_{k=1}^n\binom nk T_{\alpha,n-k}(\sigma/\alpha)\sum_{i=1}^{k} S(k,i) \alpha^{k-i} z^i \notag\\
  &=T_{\alpha,n}(\sigma/\alpha)+\sum_{i=1}^n\bigg(\sum_{k=i}^n \binom nk T_{\alpha,n-k}(\sigma/\alpha)S(k,i) \alpha^{k-i} \bigg)z^i.\label{vxzrarwaq}
\end{align}
Applying the operator $\mathcal S^{-1}$ to \eqref{vxzrarwaq} gives 
 \eqref{ThmCharlier1}.
\end{proof}

\begin{remark}In the case $\alpha=1$, the explicit form of the Charlier polynomials is well known. The reader is advised to compare formula \eqref{ThmCharlier1} in the case $\alpha=1$ with an expansion of $\binom zn$ in the Charlier polynomials discussed  in Example~1 on page 479 of  \cite{Jordan}. Corollary~\ref{fxtrstesw55u} can also be derived from \cite[Theorem~4.6]{KL} by taking the limit $\beta\to0$. 
\end{remark}

 Denote by $\partial^+$ and $\partial^-$ the raising and lowering operators for the monic polynomial sequence $(c_n)_{n=0}^\infty$, i.e., $\partial^+, \partial^-\in\mathcal L(\mathbb C[z])$ and  
\[
	(\partial^+ c_n) (z) = c_{n+1} (z), \quad (\partial^-c_n)(z) = n c_{n-1} (z), \quad\quad n \in \mathbb N_0 .
\]

The following proposition can be easily shown.

\begin{proposition} 	\label{fggh6}
The operator $\partial^-$ is closable in $L^2(\alpha\mathbb N_0,\pi_{\alpha,
\sigma})$, and let us keep the notation $\partial^-$ for its closure. Then, for each $z \in \mathbb C$, the coherent state $\mathcal E(\cdot,z)$ is eigenfunction of the operator $\partial^-$ belonging to the eigenvalue $z$.
\end{proposition}

\section{The images of the operators $\mathcal U$ and $\mathcal V$ under $\mathcal S^{-1}$}\label{fxdzdzsese}

Recall that a function $f : \alpha {\mathbb N}_0 \rightarrow \mathbb C$ belongs to $L^2 (\alpha\mathbb N_0 , \pi_{\alpha, \sigma})$ if and only if it can be uniquely represented in the form $f(\alpha k) = \sum_{n=0}^\infty f_n c_n (\alpha k)$, where $k\in\mathbb N$ and  $f_n \in \mathbb C$ ($n \in {\mathbb N}_0$) satisfy $ \lVert f \rVert^2_{L^2 (\alpha\mathbb N_0 , \pi_{\alpha, \sigma})} = \sum_{n=0}^\infty| f_n |^2  n! \, \sigma^n < \infty$. 

By \eqref{tsrear}, the Sheffer sequence $(c_n )_{n=0}^\infty$ satisfies the condition of Theorem \ref{relatedSheffer}. Hence, the statement of Corollary \ref{CorrfllwThm} holds for $s_n(z)=c_n(z)$  ($n\in\mathbb N_0$). 

 For each $l \in \mathbb N$, there exists $ C > 0$ such that, for each $f(z) = \sum_{n=0}^\infty f_n c_n (z)\in {\mathcal E}_{\min}^1(\mathbb C)$, we have $ \lVert f\rVert_{L^2 (\alpha\mathbb N_0 ,\pi_{\alpha,  \sigma})} \leq C  \vertiii{f}_{l}$.
Hence, for each $f\in {\mathcal E}_{\min}^1 (\mathbb C)$, the restriction of $f$ to $\alpha\mathbb N_0$ determines a function from $L^2 (\alpha\mathbb N_0,\pi_{\alpha,  \sigma})$, and furthermore, if the restriction of $f$ to $\alpha\mathbb N_0$ is identically equal to zero, then the function $f$ is identically equal to zero on $\mathbb C$.
Thus, we obtain an embedding of ${\mathcal E}_{\min}^1 (\mathbb C)$ into ${L^2 (\alpha\mathbb N_0 , \pi_{\alpha,  \sigma})}$ and this embedding is continuous. 

Theorem \ref{relatedSheffer}  implies

\begin{proposition} \label{poiy6yh}
The Fr\'echet space ${\mathcal E}_{\min}^1 (\mathbb C)$ is continuously embedded into ${L^2 (\alpha\mathbb N_0, \pi_{\alpha, \sigma})}$ (in the above explained sense). Furthermore, the operator $\mathcal S$ restricted to ${\mathcal E}_{\min}^1 (\mathbb C)$ is a self-homeomorphism of $ {\mathcal E}_{\min}^1 (\mathbb C)$.
\end{proposition}

Recall the linear operators $\mathcal U$ and $\mathcal V$ acting in $\mathbb C[z]$, see \eqref{cxdsarw4ry}. Define linear operators $U,V\in\mathcal L(\mathbb C[z])$ by $U=\mathcal S^{-1}\mathcal U\mathcal S$ and $V=\mathcal S^{-1}\mathcal V\mathcal S$. Thus,  
$$U = \partial^+ + \frac{\sigma}{\alpha},\quad V = \alpha  \partial^- + 1 .$$
Since $\rho = \mathcal U\mathcal V$, formula \eqref{cxsrsareare5} implies  
\begin{equation}\label{xzaewewaEW}
Z=UV.
\end{equation}

It is easy to see that the operators $U$, $V$, $Z$ and $E_h$ for $h\in\mathbb C$ admit
  (unique) extensions to continuous linear operators in $ {\mathcal E}_{\min}^1 (\mathbb C)$, for which we preserve their original notation. In particular, formula \eqref{xzaewewaEW} can be understood as an equality of continuous linear operators in $ {\mathcal E}_{\min}^1 (\mathbb C)$.

\begin{proposition}
We have the following two equalities of continuous linear operators in~$ {\mathcal E}_{\min}^1 (\mathbb C)$: 
$$U = Z  E_{-\alpha},\quad V = E_{\alpha}.$$ 
\end{proposition}

\begin{proof}By Theorems~\ref{thmBT}, \ref{thmSS} and formula \eqref{tsrear}, 
\begin{equation}\label{cfsts}
 \partial^- = C(D),
 \end{equation}
where $C(t)$ is the compositional inverse of $B(t) = \frac{1}{\alpha} \, \log (1+t\alpha)$. Hence, 
\begin{equation}\label{vcfsersreswr5w3e} C(t) = \frac{1}{\alpha} (e^{\alpha t} -1). 
\end{equation}
By  \eqref{cfsts}, \eqref{vcfsersreswr5w3e},  and Boole's formula,
$$ \partial^- = \frac{1}{\alpha} (e^{\alpha D} -1) = \frac{1}{\alpha} (E^\alpha - 1) .$$
Hence, 
$$V = \alpha  \partial^- + 1 = E^\alpha. $$ 
By \eqref{xzaewewaEW}, $Z=UE^{\alpha}$, which implies  $ U = Z  E^{-\alpha}$.
\end{proof}

		\begin{center}
	{\bf Acknowledgements} 
		\end{center}

\noindent The initial version of the paper was written when C.K. was a postdoc at Swansea University. The revision of the paper was conducted after C.K. had  become a lecturer at Suranaree University of Technology. 

C.K. was financially supported by the EPSRC, UKRI as EPSRC Maths DTP Doctoral Prize related to grant reference EP/T517987/1. 

The authors are grateful to the anonymous refree   for their helpful suggestion to improve the paper.


\begin{thebibliography}{99}

\bibitem{AW} H. Araki and E. Woods, Representations of the C.C.R. for a nonrelativistic infinite free Bose gas, {\it J. Math. Phys.} {\bf 4} (1963) 637--662.

	
\bibitem{Segal-Bargmann} N. Asai, I. Kubo and H.-H. Kuo, 
Segal--Bargmann transforms of one-mode interacting Fock spaces associated with Gaussian and Poisson measures, {\it Proc. Amer. Math. Soc.} {\bf 131} (2003) 815--823. 


\bibitem{Bargmann1} V. Bargmann, On a Hilbert space of analytic functions and an associated integral transform,
{\it Comm. Pure Appl. Math.} {\bf 14} (1961) 187--214.

\bibitem{Bargmann2} V. Bargmann, Remarks on a Hilbert space of analytic functions, {\it Proc. Nat. Acad. Sci. U.S.A.} {\bf 48} (1962) 199--204.

\bibitem{Bargmann3} V. Bargmann, Acknowledgement, {\it Proc. Nat. Acad. Sci. U.S.A.} {\bf 48} (1962) 2204.


\bibitem{BK} Y. M. Berezansky and Y. G. Kondratiev, {\it Spectral Methods in Infinite-Dimensional Analysis} (Naukova Dumka, Kyiv, 1988 (in Russian), English translation: Kluwer Academic Publishers, Dordrecht, 1995). 

\bibitem{FinkelshteinEtAll} D. Finkelshtein, Y. Kondratiev, E. Lytvynov, M. J. Oliveira and L. Streit, Sheffer homeomorphisms of spaces of entire functions in infinite dimensional analysis, {\it J. Math. Anal. Appl.} {\bf 479} (2019) 162--184.

\bibitem{GGPS} G. A. Goldin, J. Grodnik, R. T. Powers and D. H. Sharp, Nonrelativistic current algebra in the $N/V$ limit, {\it J. Math. Phys.} {\bf 15} (1974) 88--100.

\bibitem{Grabiner} S. Grabiner, Convergent expansions and bounded operators in the umbral calculus, {\it Adv. in Math.} {\bf 72} (1988) 132--167.

\bibitem{HKPS} T. Hida, H.-H. Kuo, J. Potthoff and L. Streit, {\it White Noise: An Infinite Dimensional Calculus}  (Kluwer,   Dordrecht, 1993). 

\bibitem{Jordan} C. Jordan, {\it Calculus of Finite Differences. Third Edition} (Chelsea Publishing, New York, 1965). 



\bibitem{10-Katriel} J. Katriel, Combinatorial aspects of boson algebra, {\it Lett. Nuovo Cimento} {\bf 10} (1974) 565--567. 


\bibitem{KL}  C. Kodsueb and E. Lytvynov, Segal--Bargmann transforms and generalized Weyl algebras associated with the Meixner class of orthogonal polynomials, {\it J. Math. Phys.} {\bf 66} (2025), Paper No.~123503, 22 pp.


\bibitem{Kondratiev} Y. G. Kondratiev, Spaces of entire functions of an infinite number of variables connected with the rigging of a Fock space, in {\it Spectral Analysis of Differential Operators} (Math. Inst. Acad. Sci. Ukrainian SSR, Kyiv, 1980), pp.~ 18--37 (in Russian), English translation: {\it Selecta Math. Sovietica} {\bf 10} (1991) 165--180.

\bibitem{11-Kung} J. P. S. Kung,  G.-C. Rota and C. H. Yan,  {\it Combinatorics: the Rota Way} (Cambridge University Press, Cambridge, 2009). 

\bibitem{LS} Y.-J. Lee and  H.-H. Shih, The Segal--Bargmann transform for L\'evy functionals, {\it J. Funct. Anal.} {\bf 168} (1999)  46--83.


\bibitem{12-Mansour} T. Mansour and M. Schork, {\it Commutation Relations, Normal Ordering, and Stirling Numbers} (CRC Press, Boca Raton, 2016).

\bibitem{13-Meixner} J. Meixner, Orthogonale Polynomsysteme mit einer Besonderen Gestalt der Erzeugenden Funktion, {\it J. London Math. Soc.} {\bf 9} (1934) 6--13.


\bibitem{Obata} N. Obata, {\it White Noise Calculus and Fock Space} (Springer, Berlin, 1994). 

\bibitem{Roman} S. Roman, {\it The Umbral Calculus} (Academic Press,  New York, 1984). 

\bibitem{Rota-paper}  G.-C. Rota, D. Kahaner and A. Odlyzko, On the foundations of combinatorial theory. VIII. Finite operator calculus,  {\it J. Math. Anal. Appl.} {\bf 42} (1973) 684--760.

\bibitem{Segal2}  I. E. Segal, Mathematical characterization of the physical vacuum for a linear Bose--Einstein field, {\it Illinois J. Math.} {\bf 6} (1962) 500--523. 

\bibitem{Segal1}  I. E. Segal, {\it Mathematical Problems of Relativistic Physics. Proceedings of the Summer Seminar, Boulder, Colorado, 1960, Vol. II} (American Mathematical Society, Providence, RI, 1963). 



\bibitem{Segal3} I. E. Segal, The complex-wave representation of the free boson field, in {\it Topics in Functional Analysis (Essays Dedicated to M. G. Krein on the Occasion of his 70th Birthday},  (Academic Press, New York-London, 1978), pp. 321--343.

\bibitem{Surgailis} D. Surgailis, On multiple Poisson stochastic integrals and associated Markov semigroups, {\it Probab. Math. Statist.} {\bf 3} (1984) 217--239.


\end{thebibliography}
\end{document}